\documentclass[A4paper, 12pt]{article}

\usepackage{hyperref}
\usepackage[utf8]{inputenc}
\usepackage[OT2,T1]{fontenc}
\usepackage[english,french]{babel}
\usepackage[osf]{mathpazo}
\usepackage{graphicx}
\usepackage{amsmath}
\usepackage{upgreek}

\usepackage{color}
\definecolor{FOB}{rgb}{0.,0.,0.7}
\definecolor{FOR}{rgb}{0.65,0.,0.}
\definecolor{FOfond}{rgb}{0.98,0.98,0.91}

\pagecolor{FOfond}

\usepackage{sectsty}
\allsectionsfont{\color{FOR}}

\usepackage{amssymb}
\usepackage{authblk}

\def\eqref#1{(\ref{#1})}
\def\Description#1{\relax}

\def\arctan{{\rm arctan}}

\def\dd{{\rm d}}
\def\rd{{\rm d}}

\def\cC{{\mathcal C}}

\def\kts{{\rm kts}}

\def\cO{{\mathcal O}}

\def\N{{\bf N}}

\def\R{{\bf R}}

\def\bull{\vrule height .9ex width .8ex depth -.1ex }
\def\rd{{\rm d}}

\def\rj{{\rm j}}

\def\lbb{[\![}
\def\rbb{]\!]}

\newcounter{thenum}
\def\texttheo{\relax}
\newenvironment{theorem}{\medbreak\refstepcounter{thenum}
\noindent\textsc{Theorem} %
\thethenum. \texttheo ---  \it  }{\rm }
\newenvironment{e-proposition}{\medbreak\refstepcounter{thenum}
\noindent\textsc{Proposition} \thethenum. ---  \it  }{\rm }

\newenvironment{e-definition}{\medbreak\refstepcounter{thenum}
\noindent\textsc{Definition} \thethenum. ---  \it  }{\rm }

\newenvironment{e-rem}{\medbreak\refstepcounter{thenum}{}%
 \thethenum) }{}

\newenvironment{e-ex}{\medbreak\refstepcounter{thenum}{}%
 \thethenum) }{}
\newenvironment{proof}{\smallbreak\noindent{\sc Proof.} --- \rm}{\quad\bull\smallskip\rm}

\newenvironment{definition}{\medbreak\refstepcounter{thenum}
\noindent\textsc{Définition} \thethenum. ---  \it  }{\rm }

\begin{document}

\thispagestyle{empty}

\begin{center}{\LARGE\parindent=0pt{\color{FOR}
Extending Flat Motion Planning to Non-flat Systems.\\
  Experiments on Aircraft Models Using Maple

}}
\end{center}
\vskip2cm

\hbox to \hsize{\parindent =0pt\hbox to 2.5cm{\hfill}\hss\vbox{\hsize
    = 7cm {\large François \textsc{Ollivier}} 
\bigskip

LIX, UMR CNRS 7161 

École polytechnique 

91128 Palaiseau \textsc{cedex}

France

\smallskip

{\tiny francois.ollivier@lix.polytechnique.fr}
} \hss}
\vskip 0.3cm

\begin{center}\parindent =0pt May, 12$^{\rm th}$  2022
\end{center}
\vfill

{\small

\noindent \textbf{Abstract.} 
Aircraft models may be considered as flat if one neglects some terms
associated to aerodynamics. Computational experiments in Maple show
that in some cases a suitably designed feed-back allows to follow such
trajectories, when applied to the non-flat model. However some
maneuvers may be hard or even impossible to achieve with this flat
approximation.  In this paper, we propose an iterated process to
compute a more achievable trajectory, starting from the flat reference
trajectory. More precisely, the unknown neglected terms in the flat
model are iteratively re-evaluated using the values obtained at the
previous step. This process may be interpreted as a new trajectory
parametrization, using an infinite number of derivatives, a property
that may be called \emph{generalized flatness}.  We illustrate the
pertinence of this approach in flight conditions of increasing
difficulties, from single engine flight, to aileron roll.
  \smallskip


\vfill\eject

\thispagestyle{empty}

\noindent \textbf{Résumé.}  Des modèles d'avions peuvent être
considérés comme plats si on néglige certains termes associés à
l'aérodynamique. Des expériences de calcul en Maple montrent que dans
certain cas, un bouclage convenable permet de suivre de telles
trajectoires, en utilisant le modèle non plat. Certaines manœuvres
peuvent néanmoins être difficiles, voire impossible à réaliser avec
cette approximation plate. Dans cet article, nous proposons un
processus itératif pour calculer une trajectoire plus aisée à suivre,
en commençant par l'approximation plate de référence. Plus
précisément, les termes inconnus négligés dans le modèle plat, sont
itérativement réévalués, en utilisant les valeurs obtenues à l'étape
précédente. Ce processus peut être interprété comme un nouveau
paramétrage utilisant une infinité de dérivées, une propriété qui peut
être appelée \emph{platitude généralisée}. Nous illustrons la
pertinence de cette approche dans des conditions de vol de difficulté
croissante, incluant un vol avec un seul moteur, une descente en vol
plané avec glissade et une manœuvre de voltige.
\smallskip

  \noindent Mots-clés: systèmes plats, planification de trajectoire,
  contrôle de vol, opérateur de Newton, calcul symbolique-numérique,
  platitude généralisée.

}

\vfill\eject

\selectlanguage{english}

\section*{Introduction}

We illustrate the use of computer algebra for experimental
investigations relying on numerical simulations in the field of
automatic control. We consider here the notion of flat systems and
some possible generalizations in order to solve motion planning
problems for aircrafts models.

The solutions of flat systems~\cite{FLMR95,FLMR99,Levine09} can be
parametrized by a set of functions, called \textit{flat outputs}, and
a finite number of their derivatives. This property is particularly
useful for motion planning of non-linear systems, \textit{i.e.} the
design of a control law able to generate a trajectory joining a given
starting point to a given end point. Though flatness is not a generic
property, flat systems are ubiquitous in practice. There is no known
complete algorithm to decide flatness (see \textit{e.g.}
Lévine~\cite{Levine11} for necessary and sufficient conditions), but
the flat outputs have often simple expressions that may be guessed by
physical considerations.

This work takes place in a systematic study of \textit{apparent
  singularities} of flat systems, \textit{i.e.} points where the
parametrization provided by given flat outputs ceases to be
defined~\cite{KLO18,KLO20}.  In practice, such situations are more
likely to appear when a failure modifies the symmetries of the system
or involves the loss of some controls, thus requiring an alternative
flat output.

Among the classical examples of flat systems are cars, trucks with
trailers, cranes, aircrafts, etc. Note that aircraft models have been
studied since long in~\cite{MartinPHD,Martin1996}. Although
aerodynamics models are complex and may involve many parameters, they
turn out to be flat if one neglects the thrust created by control
surfaces (rudder, elevator and ailerons) or associated to angular
speeds, a legitimate approximation in many cases.

In practice, we aim at designing a suitable feed-back able to
compensate both perturbations and modelling errors.  In order to
investigate its robustness in the context of maneuvers and failures of
increasing difficulties, we have designed a package in Maple.  Its
implementation is presented and we illustrate its use by a few
numerical simulations of trajectory tracking. More details will be
given in a forthcoming papers with Y.J.~Kaminski.

We focus here on a notion of \emph{generalized
  flatness}, suggested by computational experiments, trying
to improve trajectory tracking when the design of
a suitable feed-back becomes hard. We first noticed that, considering
trajectories with constant controls and attitude angles, these
controls and angles may be computed by solving an algebraic system,
\textit{i.e.} a non-differential one. The real model is in this case
more complicated, but of the same nature as the simplified one. We
sometimes needed to use an alternative simplified model, where control
values are not set to $0$ but to constant values provided by
\textit{ad hoc} calibration functions.

We tried then to go further and to improve the parametrization
provided by the simplified model. We have needed to neglect some
terms, depending on the controls $U$. As the flat parametrization
provides a first evaluation $U^{[0]}$ for the controls, we can use
this value in the perturbation terms of the full model, instead of
setting them to $0$. We get so a second evaluation $U^{[1]}$ for the
controls that may be used to improve the evaluation of the
perturbation terms, providing a third evaluation $U^{[2]}$\dots\ This
process can be iterated \textit{ad libitum}. In our experiments, this
simple change provides, using only $4$ iterations, a precise motion
planning for the full aerodynamic model, which suggests the
introduction of a notion of \emph{generalized flatness} for such
systems. ``Precise'' means here that the trajectories remain close to
the values of the flat outputs, without using any feed-back. See
simulations in sec.~\ref{sec:examples}.  As each iteration implies
more derivatives of the flat outputs, such a generalized flat
parametrization potentially involves an \textit{infinite} number of
derivatives of the flat outputs of the unperturbed flat system.

Flat systems and their singularities are introduced in sec.~\ref{sec:flat}.
Detailed aircraft models, for which this motion planning algorithm has
been taylored, are presented in sec.~\ref{sec:models} and
their approximate flatness and singularities are studied in
sec.~\ref{sec:flat-outputs}. Then, their motion planning,
tracking feed-back and the associated Maple package are presented in
sec.~\ref{sec:Maple-pack}, the implementation of generalized flatness
in section~\ref{sec:generalized}, followed by examples of flight
maneuvers with increasing difficulties in
section~\ref{sec:examples}. A last section~\ref{sec:theory},
provides preliminary elements for a theoretical interpretation.
\vfill\eject

\section{Flat systems and their singularities}\label{sec:flat}

The first definition of flatness was given in the framework of
differential algebra~\cite{Ritt50}. We prefer here to use a more
flexible definition, relying on Vinogradov's notion of
diffieties~\cite{Vinogradov1986,Zharinov1996}, that do not restrict to
algebraic systems and algebraic flat outputs. The main difference in
our approach, is that diffieties are defined by fixing a derivation,
which corresponds to flatness and not just the distribution generated
by the associated vector field, which corresponds to orbital flatness
when time scaling is allowed. See~\cite{FLMR99}.

\subsection{Definition}

\begin{definition} A diffiety $V$ is an open\footnote{Using the
    coarsest topology that makes the $i^{\rm th}$ projection map
    $\pi_{i}$ continuous, for all $i\in I$.} subset of $\R^{I}$, where
  $I$ is a denumerable set, equipped with a derivation $\delta$. All
  functions on a diffiety are $\cC^{\infty}$ and only depend on a
  \emph{finite} number of coordinates. We denote their set by
  $\cO(V)$.
\end{definition}

In the sequel, we will be concerned with diffieties associated to a
system of finitely many ordinary differential equations
\begin{equation}\label{eq:sys}
  x_{i}'=f_{i}(x,u,t),
\end{equation}
where $x=(x_{1}, \ldots,
x_{n})$ is the \textit{state vector}, $u=(u_{1}, \ldots, u_{m})$ the
\textit{controls} and $t$ is the time, implicitly satisfying
$t'=1$. To such a system, we associate
$\R\times\R^{n}\times\left(\R^{\N}\right)^{m}$, the first copy of $\R$ is for
$t$, then $\R^{n}$ for $x$ and the last term corresponds to the
controls and their derivatives. So the derivation $\delta$, that we
denote by $\dd_{t}$ is
\begin{equation}\label{eq:sys-diff}
\dd_{t}:=\partial_{t}+\sum_{i=1}^{n}f_{i}(x,u,t)\partial_{x_{i}}+\sum_{j=1}^{m}\sum_{k\in\N} u_{j}^{(k+1)}\partial_{u_{j}^{(k)}},
\end{equation}
denoting $\partial/\partial x_{i}$ by $\partial_{x_{i}}$ for
simplicity. We may obviously restrict to an open subset, according to
physical limitations.

Among such diffieties, is the \textit{trivial diffiety}, which is
$\R\times(\R^{\N})^{m}$, equipped with 
$$
\dd_{t}:=\partial_{t}+\sum_{j=1}^{m}\sum_{k\in\N} z_{j}^{(k+1)}\partial_{z_{j}^{(k)}},
$$
which is in fact the jet space $J^{\infty}(\R,\R^{m})$.
We are now able to define flatness.

\begin{definition}
A diffiety morphism $\phi:U_{\delta_{1}}\mapsto V_{\delta_{2}}$ is
such
that $\phi^{\ast}\cO(V)\subset\cO(U)$ and, for any function $g$ on $V$,
$\phi^{\ast}\delta_{2}g=\delta_{1}\phi^{\ast}g$, meaning that the mapping $g$ is
compatible with the derivations. 

  The \emph{flatness domain}, is the set all \emph{flat
    points}, \emph{i.e.} points admitting a neighborhood
  isomorphic to an open subset of the trivial diffiety.

  Let $\phi$ be such an isomorphism defined by $z_{j}:=Z_{j}(x,u,t)$, the
  functions $Z_{j}$ are called \emph{flat outputs}.

Thus, $\phi^{-1}$ is locally defined and provides a \emph{flat
parametrization}, defined by $x_{i}=X_{i}\left(z, \ldots, z^{(r)}\right)$ and
$u_{j}^{(k)}= U_{j,k}\left(z, \ldots, z^{(r+k+1)}\right)$.
\end{definition}

In many cases, the state space is not affine and can be a
sphere, a circle\dots\ as we will see soon. In such cases, different
charts need to be used to cover it. And flatness can impose to use
more charts, each associated to a suitable flat output, in order to
cover the whole flatness domain.

\subsection{Singularities of flat systems}

In the above definition, flat outputs are only defined on open
spaces. Points where flat outputs are not defined, or the inverse
mapping, are \emph{apparent singularities} for these
outputs. \emph{Flat singularities} are the points where no flat
parametrization can be defined.

The lack of a general algorithmic criterion to decide flatness makes
difficult to characterize flat singularities. In a first stage of
a collaboration in progress with Y.J.~Kaminski and J.~Lévine, we have
focused on driftless systems~\cite{KLO18} and affine
systems~\cite{KLO20} with $n-1$ controls, for which the following
necessary condition, which amounts to the controllability of the
linearized system, turns out to cover all the cases when the action of
the control functions remain independent.

The most precise expression of this criterion requires using power
series. At a given point $\eta$ of a diffiety, we associate to any
function $F$ the power series:
$\rj_{\eta}F:=\sum_{k\in\N}\rd_{t}^{k}F(\eta)t^{k}/k!$ and consider at
each point $\eta$ the differential operator
\begin{equation}
  \dd_{\eta} F:= \sum_{k=1}^{n}\rj_{\eta}(\partial_{x_{k}}f_{i})\dd  x_{k}
   +\sum_{j=1}^{m}\sum_{k\in\N}\rj_{\eta}(\partial_{u_{j}^{(k)}}f_{i})\dd u_{j}^{(k)}.
\end{equation}

\begin{theorem}\label{th:lin-crit} If a diffiety defined by a
  differential system \eqref{eq:sys} is flat at point $\eta$, then the
  $\R\lbb t\rbb[\dd_{t}]$-module defined the \emph{linearized system
    at $\eta$} $\dd_{\eta}(x_{i}'-f_{i}(x,u,t)$, that is the quotient
  $\R\lbb t\rbb[\dd_{t}]$-module
  $(\dd_{\eta}x,\dd_{\eta}u)/(\dd_{\eta}(x_{i}'-f_{i}(x,u,t)))$, is a
  free module.
\end{theorem}
\begin{proof} If $Z$ is
  a flat output, then $\dd_{\eta} Z$ is a basis of this module. Indeed,
  $x_{i}=X_{i}(Z)$, for $1\le i\le n$ and 
  $u_{j}=U_{j}(Z)$, for $1\le j\le n$, so that
  $\dd_{\eta}x_{i}=\dd_{\eta}X_{i}(Z)$ and
  $\dd_{\eta}u_{j}=\dd_{\eta}U_{j}(Z)$. 
\end{proof}

It seems that we are lacking a good reference for testing freeness
of a $D$-module with coefficient in a power series ring. But things
are easy when coefficients are constants.

\section{Aerodynamic models of aircrafts}\label{sec:models}

We have used the model described by Martin~\cite{MartinPHD,Martin1996}
that basically follows most textbooks. We avoid reproducing all lengthy
equations to focus on their structure. 

It is classical to model aircrafts using the following $12$ state
variables: $(x,y,z,V,\gamma,\chi,\alpha,\beta,\mu,p,q,r)$. We try to
describe briefly their rough meaning. A precise understanding is not
mandatory for what follows. First, $(x,y,z)$ correspond to the
coordinates of the gravity center of the aircraft, $V$ to its speed,
the flight path angle $\gamma$ and the azimuth angle $\chi$ are Euler
angles describing the speed vector, $\mu$ is the bank angle,
corresponding to roll. Those three Euler angles define the \textit{wind
  frame}, and the sideslip angle $\beta$ together with the angle of
attack $\alpha$ describe respectively the rotations with respect to
the $z$-axis (yaw) and then $y$-axis (pitch) in order to go from the
wind referential to the \textit{aircraft frame}, according to the
following figure.
\begin{figure}[h]
  \centering
  \includegraphics[width=5cm]{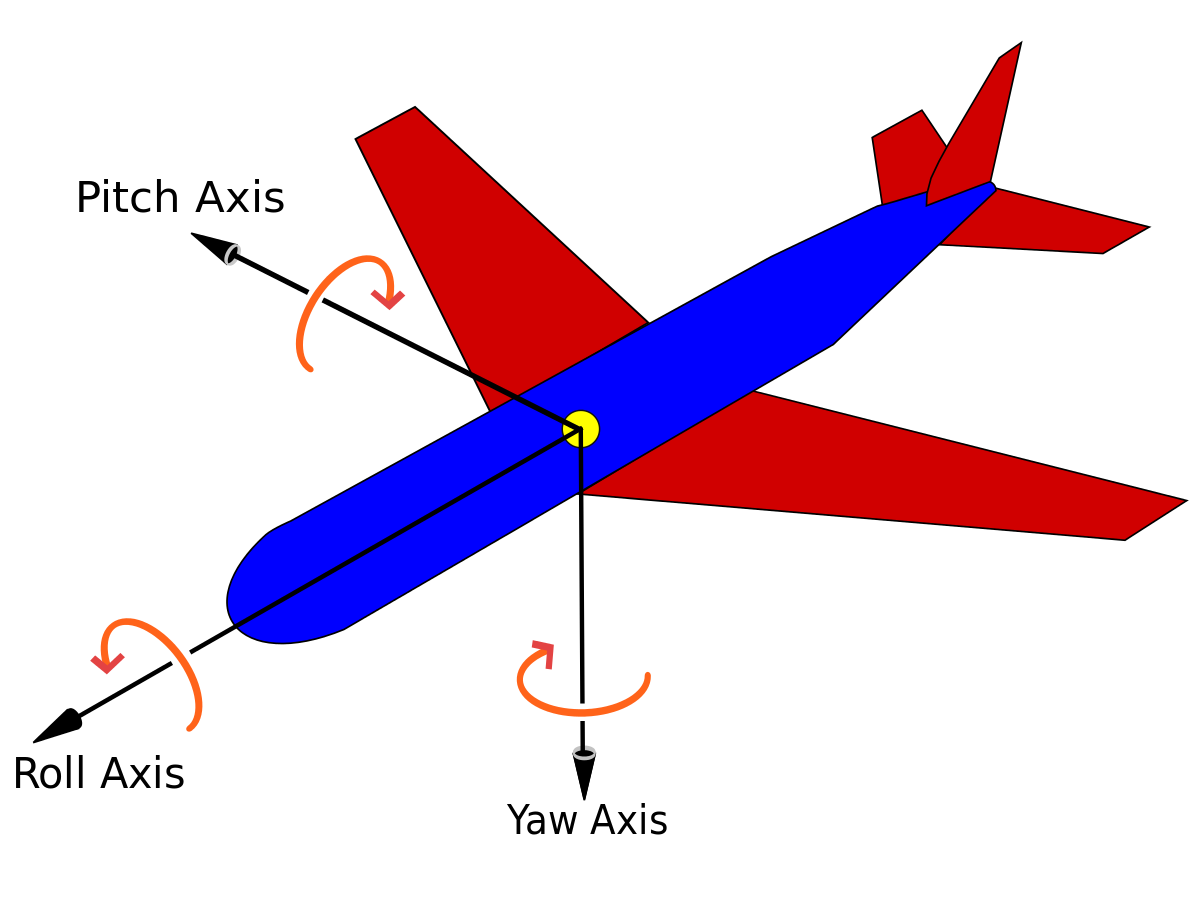}
  \hbox to 0pt{\tiny Thanks to Wikipedia\hss}\\
  \vbox{Angle $\mu$ corresponds to roll, $\beta$ to yaw and
    $\alpha$ to pitch.}
    \caption{Aircraft rotation axes}
  \Description{Angle $\mu$ corresponds to roll, $\beta$ to yaw and
    $\alpha$ to pitch. }
\end{figure}

Then, $(p,q,r)$ is the expression of the
rotation vector in the Galilean referential tangent to the aircraft
referential at each time.

The controls are the following, the thrust of both engines
$(F_{1},F_{2})$, that we prefer to model using their sum
$F=F_{1}+F_{2}$ and a parameter $\eta:=(F_{1}-F_{2})/(F_{1}+F_{2})$,
and then the virtual angles $\delta_{\ell}$, $\delta_{m}$ and $\delta_{n}$, that
respectively express the positions of the ailerons, elevators and rudder. 
When the rudder is damaged, it is possible to some extent to use
differential thrust $\eta$ as a control instead of $\delta_{n}$ (see,
\textit{e.g.} \cite{Lu-Turgolu-2018}). 

\subsection{The shape of the equations}\label{subsec:shape-eq}

We can now describe the shape of the equations, dividing the state
variables in $4$ subsets: $\Xi_{1}:=\{x,y,z\}$, $\Xi_{2}:=\{V,\gamma,\chi\}$,
  $\Xi_{3}:=\{\alpha,\beta,\mu\}$ and $\Xi_{4}:=\{p,q,r\}$. We have:
\begin{subequations}
\begin{align}  \label{eq:shape1}
      (x',y',z') &= G_{1}(V,\gamma,\chi);\\
      \label{eq:shape2}
      (V',\gamma',\chi')&=G_{2}(V,\gamma,\alpha,\beta,\mu,F,[p,q,r,\delta_{\ell},\delta_{l},\delta_{n}]);\\
      \label{eq:shape3}
      (\alpha',\beta',\mu')&=G_{3}(V,\gamma,\alpha,\beta,\mu,p,q,r);\\
      \label{eq:shape4}
      (p',q',r')&=G_{4}(V,\gamma,\alpha,\beta,\mu,p,q,r,\delta_{\ell},\delta_{l},\delta_{n}).
\end{align}
\end{subequations}
The equation \eqref{eq:shape2} actually depends on
$p,q,r,\delta_{\ell},\delta_{l},\delta_{n}$, but this dependence is
often neglected. With this simplification, setting
$\Xi_{5}:=\{\delta_{\ell},\delta_{l},\delta_{n}\}$, at stage $i$, we
can generically express the value of $\Xi_{i+1}$, using the
derivatives $\Xi_{i}'$. At stage $2$, \textit{i.e.} for $i=2$, we
need to choose one variable $\zeta$ in the set
$\Xi_{3}=\{\alpha,\beta,\mu,F\}$ to form a flat output. Then, generically,
$x,y,z,\zeta$ and their derivatives allow to compute the values of the
state space and controls. The classical choice is $\zeta=\beta$.  We now
briefly investigate apparent singularities that may appear at each
level of derivation, the second one being left for further
investigations.
    
\subsubsection{Stage 1}
{\small
\begin{subequations}
\begin{align}
\label{eq::aircraft1:x}
  \frac{d}{d t} x(t)  & =  V(t) \cos{\left(\chi(t) \right)} \cos{\left(\gamma(t) \right)};\\
\label{eq::aircraft1:y}
\frac{d}{d t} y(t)  & =  V(t) \sin{\left(\chi(t) \right)} \cos{\left(\gamma(t) \right)}; \\
\label{eq::aircraft1:z}
\frac{d}{d t} z(t)  & =  - V(t) \sin{\left(\gamma(t) \right)}.
\end{align}
\end{subequations}

}

It is easily seen that the values of $V$, $\chi$ and $\gamma$, modulo
$\pi$, can be
computed, provided that $V\cos(\gamma)\neq0$, which
seems granted in most situations. The vanishing of $V$ may occur
with aircrafts equipped with vectorial thrust, which means a larger set
of controls, that we won't consider here. This means that we assume
$V>0$, so that a single value for $(\cos(\chi),\sin(\chi))$ can be
determined on the unit circle. The vanishing of
$\cos(\gamma)$ can occur with loopings etc.\ and would require the
choice of a second chart with another set of Euler angles. This issue
was not investigated here.

\subsubsection{Stage 3}

We postpone the study of stage 2, that contains the main
difficulties, to the next section.
The shape of the third level equations imposes $\cos(\beta)\neq0$. They are
linear in $(p,q,r)$, with a non vanishing determinant and so easily solved.

\subsubsection{Stage 4}
The case of variables $(p,q,r)$ is easy too.

The dynamics of the angular speed matrix $(p,q,r)$ is given by:

{\small
\begin{equation}
\label{eq::aircraft2}
\left(\begin{array}{l}\frac{d}{d t} p(t)\\
  \frac{d}{d t} q(t)\\
  \frac{d}{d t} r(t)
\end{array}\right) =
I^{-1}
\left(\begin{array}{l}
 (I_{yy}-I{zz})qr+I_{xz}pq+L\\
 (I_{zz}-I_{xx})pr+I_{xz}(r^{2}-p^{2})+M\\
  (I_{xx}-I_{yy})pq-I_{xz}rq+N
\end{array}\right),
\end{equation}

}

where $I$ is the inertia matrix of the aircraft, assumed to be symmetric
with respect to the $xz$-plane, and $(L,M,N)$ the torque, that can
obviously be computed using these equations. In general,
one expects $L$ to depend mostly of $\delta_{\ell}$, $M$ on
$\delta_{m}$, etc.\ and to be monotonous in the range of admissible
values. Using the GNA model, they are linear in those controls, with
invertible matrices.

\subsection{The GNA model}

The aircraft model equations involve the forces $(X,Y,Z)$ and the torques
$(L,M,N)$ acting on the aircraft, that are given by these formulas:

{\small
\begin{subequations}
\begin{align}
\label{eq::X}
X&=F(t)\cos(\alpha+\epsilon)\cos(\beta(t))-\frac{\rho}{2}SV(t)^{2}C_{x}- gm\sin{\left(\gamma(t) \right)};\\
\label{eq::Y}
Y&=\begin{array}[t]{l}-F(t)\cos(\alpha+\epsilon)\sin(\beta(t))+\frac{\rho}{2}SV(t)^{2}C_{y}\\
   +gm\cos(\gamma(t))\sin(\mu(t));\end{array}\\
\label{eq::Z}
Z&=-F\sin(\alpha+\epsilon)-\frac{\rho}{2}SV(t)^{2}C_{z}-gm\cos(\gamma(t))\cos(\mu(t));\\
\label{eq::L}
L&=-y_{p}\sin(\epsilon)(F_{1}(t)-F_{2}(t))+\frac{\rho}{2}SV(t)^{2}aC_{l};\\
\label{eq::M}
M&=\frac{\rho}{2}SV(t)^{2}bC_{m};\\
\label{eq::N}
N&=y_{p}\cos(\epsilon)(F_{1}(t)-F_{2}(t))+\frac{\rho}{2}SV(t)^{2}aC_{n}.
\end{align}
\end{subequations}
}

The angle $\epsilon$ is a small angle related to the lack of
parallelism of the reactors with respect to the $xy$-plane of the
aircraft and $\rho$
is the air density, $a$ and $b$ lengths related to the aircraft characteristics.

The aerodynamic coefficients $C_x,C_y,C_z,C_{l},C_{m},C_{n}$ depend on
$\alpha$ and $\beta$ and also on the angular speeds $p$, $q$, $r$
as well as the controls $\delta_{l}$, $\delta_{m}$ and
$\delta_{n}$. To make the system flat, we need to consider that $C_{x}$,
$C_{y}$ and $C_{z}$ only depend on $\alpha$ and $\beta$.
In the literature, the available expressions are often partial or
limited to linear approximations, as in McLean~\cite{McLean}.
We used here the Generic Nonlinear Aerodynamic (GNA)
subsonic models, given by Grauer
and Morelli~\cite{Grauer-Morelli2014}, that cover a wider range of values.

We will provide simulations with 2 aircrafts among the $8$ in their
database: STOL utility aircraft DHC-6 Twin Otter and
the sub-scale model of a transport aircraft GTM (see~\cite{Hueschen}).
Data for the F4 and F16C fighters are also available in our
implementation. 
The GNA model for the aerodynamics functions $C$ appearing in
formulas~(\ref{eq::X}--\ref{eq::N}) depends on 45 aerodynamic
coefficients, in formulas such as:
\begin{equation}\label{eq::GNA}
{\scriptstyle \begin{array}{ll}
  C_{D} &=\begin{array}[t]{l}\theta_{1}+\theta_{2}\alpha+\theta_{3}\alpha\tilde q
    +\theta_{4}\alpha\delta_{m}+\theta_{5}\alpha^{2}
    +\theta_{6}\alpha^{2}\tilde q+\theta_{7}\delta_{m}
    +\theta_{8}\alpha^{3}\\+\theta_{9}\alpha^{3}\tilde q
    +\theta_{10}\alpha^{4},\end{array}\\
  C_{y} &=\theta_{11}\beta+\theta_{12}\tilde p
    +\theta_{13}\tilde r +\theta_{14}\delta_{l}+\theta_{15}\delta_{n},\\
  C_{L} &=\theta_{16}+\theta_{17}\alpha+\theta_{18}\tilde q
    +\theta_{19}\delta_{n}+\theta_{20}\alpha\tilde q+\theta_{21}\alpha^{2}
    +\theta_{22}\alpha^{3}+\theta_{23}\alpha^{4},
\end{array}}
\end{equation}
where $\tilde p=ap$, $\tilde r=ar$, $\tilde q=bq$, $a$ and $b$ being
constants related to the aircraft geometry, $C_{D}$ and $C_{L}$
correspond to the lift and drag coefficients. The coefficients $C_{x}$
and $C_{z}$ in the wind frame are then given by the formulas:
\begin{equation}\label{eq::LDtoXY}
\begin{array}{lll}
C_{x}&=&\cos(\alpha)C_{D}+\sin(\alpha)C_{L},\\
C_{z}&=&\cos(\alpha)C_{L}-\sin(\alpha)C_{D}.
\end{array}
\end{equation}

Grauer and Morelli also provide all the needed physical constants, but
no precise data for landing conditions, flaps\dots\ To simulate
landing, empirical changes were made. The starting point of this work
was to be able to handle the full model, considering changes of flat
outputs when singularities are met, and to question the validity of
the motion planning provided by a flat simplified model, when trying
to control the full model.

\section{Flat outputs and their singularities}\label{sec:flat-outputs}

We now investigate the singularities related to the various choices of
flat output, at stage two.

\subsection{Classical flat outputs}\label{subsec:classical-FO}

Martin~\cite{MartinPHD} has used the set of flat outputs:
$x,y,z,\beta$.  We need to explicit under which condition such a flat
output is non singular. The differential equations involved at
stage two are the following.  {\small
\begin{subequations}
  \begin{align}
    \label{eq::aircraft1:V}
\frac{d}{d t} V(t) & =  \frac{X}{m} ;\\
\label{eq::aircraft1:gamma}
\frac{d}{d t} \gamma(t)& = 
                    -\frac{Y\sin(\mu(t))+Z\cos(\mu(t))}{mV(t)};\\
\label{eq::aircraft1:chi}
\frac{d}{d t} \chi(t)& =
              \frac{Y\cos(\mu(t))-Z\sin(\mu(t))}{\cos(\gamma(t))mV(t)}.
\end{align}
\end{subequations}

}
The first one~\eqref{eq::aircraft1:V} provides the value of $X$. From its
expression,  we can express the value of $F$ by~\eqref{eq::X}, as
$\alpha+\epsilon$ is assumed to be small. We see that the two last
equations depend on
$\cos(\mu)Y-\sin(\mu)Z$ and $\sin(\mu)Y
+\cos(\mu)Z$. We get new expressions $\hat Y$ and $\hat Z$ by
substituting in them the value of $F$ provided
by~\eqref{eq::X}. 
We can compute locally $\alpha$ and $\mu$, provided that
\begin{equation}\label{eq::cond_beta}
\left|\begin{array}{cc}\frac{\partial\hat X}{\partial\alpha}&
\frac{\partial\hat X}{\partial\mu}\\
\frac{\partial\hat Y}{\partial\alpha}&
\frac{\partial\hat Y}{\partial\mu}\end{array}\right|\neq0.
\end{equation}
This condition implies that $Y$ and $Z$ do not both vanish, which
excludes $0$-g flight for space training or some aerobatics maneuvers,
but which stands in most usual flight conditions. The main interest of
this choice is to be able to impose $\beta=0$, which is almost always
required.

\subsection{The bank angle choice}\label{subsec:Bank-FO}

Considering the flat
output $\{x,y,z,\mu\}$, we see that we can compute the values of
$X$, $Y$ and $Z$. Again, $X$ provides an expression of $F$, that may
be susbsituted in $Y$ and $Z$ to get new expressions $\tilde Y$ and
$\tilde Z$. The flat output is regular when
\begin{equation}\label{eq::stalling_cond}
\left|\begin{array}{cc}\frac{\partial\tilde Z}{\partial\alpha}&
\frac{\partial\tilde Z}{\partial\beta}\\
\frac{\partial\tilde Y}{\partial\alpha}&
\frac{\partial\tilde Y}{\partial\beta}\end{array}\right|\neq0.
\end{equation}
The vanishing of this determinant may be interpreted as some kind of
stalling condition. Indeed, when $\beta=0$, it is equal by symmetry to
$\partial\tilde Z/\partial\alpha \partial\tilde Y/\partial\beta$. For
most aircrafts, $\partial\tilde Y/\partial\beta\neq0$ seems reasonable,
although it may be very small or even negative for some fighter like
the F16XL with a delta wing, according to data in
\cite{Grauer-Morelli2014}. Then, $\partial\tilde Z/\partial\alpha$
means that the lift is extremal, which may be taken as a rough
mathematical definition of stalling. Of course, we are here working
with a simplified model that cannot reflect the irreversible changes
in air flow that occurs in real stalling, but only mimics it as a maximum
of the lift. In such a situation, the control $\delta_{m}$ that acts
on $\alpha$, and so on the lift, may be considered as lost.  And indeed,
for a straight line trajectory with constant speed equal to the
stalling speed, \textit{i.e.} with $\alpha$ maximal, the aircraft
model is not flat according to th.~\ref{th:lin-crit}. This means
that such a flight output always works for most aircrafts, except in
situations that obviously need to be avoided for safety reasons.

\subsection{The thrust choice}\label{subsec:Thust-FO}

The choice of thrust $F$ has one main interest: to set $F=0$ and
consider the case of an aircraft having lost all its engines. See
subsection~\ref{subsec:forward_slip}. In the case of the GNA model,
$C_{y}$ is linear in $\beta$. If $\cos(\mu)\theta_{11}\neq0$ (see
\eqref{eq::GNA}), we may express $\beta$ depending on $\alpha$, $\mu$,
$X_{1}$, $X_{2}$ and the aircraft parameters, using equation
\eqref{eq::aircraft1:chi}, and then replace it by this evaluation in
$X$ and $Z$ to get new expressions $\bar X$ and $\bar Y$. The flat
outputs including $F$ are non singular iff
\begin{equation}\label{eq::cond_F}
\left|\begin{array}{cc}\frac{\partial\bar X}{\partial\alpha}&
\frac{\partial\bar X}{\partial\mu}\\
\frac{\partial\bar Z}{\partial\alpha}&
\frac{\partial\bar Z}{\partial\mu}\end{array}\right|\neq0.
\end{equation}
By symmetry, both $\partial\bar X/\partial\mu$ and $\partial\bar
Z/\partial\mu$ vanish when $\beta=\mu=0$, so that this choice of
linearizing outputs requires non zero side-slip angle and bank angle
for trajectories included in a vertical plane.

\subsection{Other sets of flat outputs}\label{subsec:Other-FO}

Among the other possible choices for completing the
set $\Xi_{1}$ in order to get flat outputs, $\alpha$ could work in
theory but does not
seem to have much specific interest. One may also consider time varying
expressions, e.g. linear combinations of $\beta$ and $\mu$, to
smoothly go from one choice to another, which has been implemented
but did not lead to a convincing use in simulations.

\section{Maple package}\label{sec:Maple-pack}
We describe here an experimental implementation, only designed at this
stage for our own use and lacking of documentation and
comments. However, the source code is made available for curious
readers:\hfill\break
\url{http://www.lix.polytechnique.fr/{\textasciitilde}ollivier/GFLAT/}. The
goal was to get reliable results by minimizing the needed total amount
of time, that is the time requested by numerous simulations and the
time of implementation. 

Four Maple packages were written. The package \texttt{GNA} implements
data from Grauer and Morelli, the package \texttt{Flat\_Plane\_G2}
implements the flat motion planning and its generalization. A package
\texttt{Newton} contains a multivariable Newton method and a package
\texttt{Display\_plane} deals with numerical simulations and drawing
the curves that illustrate this paper.

Another important point is to be able to control long computations in
order to stop them if something goes wrong. The functions were mostly
used in verbose mode, displaying the index $i$ of each new time step
or intermediate numerical results during motion planning or numerical
integration.

This proved important for debugging but also during the repeated trial and
error sequences required to guess working parameters for the feed-back.

The general spirit was to limit ourselves to basic Maple functions:
manipulation of lists, substitutions, computation with polynomials and
classical functions, power series, the solve function for linear
systems, and the \texttt{dsolve} numerical integrator.

\subsection{Physical models. GNA}

There is not much to say about this package. Our choice was to use
global variables to store all the requested parameters. It has many
drawbacks, including some possible protest from Maple numerical
integrator, that we were able to overcome. The main advantage is to
alleviate the number of arguments in functions that already
require a great number of them and to make all the requested
intermediate results available for the function used at next step
without mistakes and omissions. There is a function for each model of
aircraft that store the physical constants, with names such as
\texttt{TO}, \texttt{GTM} or \texttt{F16C}. Its arguments are of the form
\texttt{x = fct(t), y = fct(t)}\dots a sequence that is stored in a global
list to provide the time functions associated to the flat outputs.

We have already said that any combination $\zeta$ of $\beta$ and $\mu$ can be
used as a flat output. For this, the syntax
$$
\texttt{zeta=(f1(t)*betta+f2(t)*mu=fct(t))}
$$ is recognized. One may notice that \texttt{beta} and \texttt{gamma}
are already used by Maple. An ugly but fast solution was to write
\texttt{bbeta} and \texttt{gama} to avoid conflicts. In case of rudder
failure, one can use relative thrust as control. A generic name for
this control is \texttt{u\_4} and one may write \textit{e.g.}
\texttt{deltan=u\_4, eta=0} or \texttt{deltan=10*deg, eta=u\_4}. If a
non zero value is given to $\delta_{n}$, it will be used at the
stage~2 (see~\ref{subsec:shape-eq}), for better precision, instead of
setting it to $0$ in order to define the simplified model. Options
provide models for ground effect or an expression of air density,
depending on altitude. For this, the notation \texttt{\_z} is used
instead of \texttt{z} to avoid too early evaluation. One may also
assign to \texttt{eta} a function of the time, e.g. to model an engine
failure. We need then to denote the time by \texttt{\_t}, again to
prevent too early evaluation.

\subsection{Newton operator with series}

The main task of the Flat-plane model is to achieve motion
planning. Following the ideas developed in
section~\ref{sec:flat-outputs}, this is in principle easy. We
encounter two difficulties. First, computing successive derivatives of
the flat outputs may lead to formulas of great size and slow
computations, mostly when trying to model complicated maneuvers and
long flight sequences. Second, we cannot rely at stage 2
(see~\ref{subsec:shape-eq}) on closed form formulas for solving the
equations, so that numerical approximations need to be computed.

Our choice was to compute at a given time a power series expansion of
the flat outputs $(x,y,z)$ with all terms up to $t^{5}$. At stage $2$,
a classical Newton method is used to compute constant terms of the
series corresponding to $\alpha$, $\beta$, $\mu$ or $F$. Then, we use
a Newton method for series (see \textit{e.g.} \cite[th.~3.12
  p.~70]{aecf-2017-livre}) to compute their power series expansion
modulo $t^{2}$ and then $t^{4}$, which is enough to get
$\delta_{\ell}$, $\delta_{m}$ and $\delta_{n}$ as affine functions of
$t$. Higher orders may also be computed and will be needed in
sec.~\ref{sec:generalized}.

Unless physical considerations makes it difficult or impossible
(\textit{e.g.} near stalling conditions), the use of Newton method is
in general easy, when initiated with $0$ values, as most angles are
small. This is no longer the case with flat outputs $x$, $y$, $z$ and
$F$, that require higher values of $\beta$. Then, some calibration
functions (see~\ref{subsec:calibration}) are used to provide suitable
values to initiate the computation. Our Newton function is a memory
one, so that it starts at step $i+1$ with the values of step $i$ for
better efficiency. During experiments, warning messages from the
Newton function that fails to provide solution up to $10^{-3}$ after
$20$ iterations are the symptoms of a choice of trajectory that is too
close to a singularity of the flat output.

\subsection{Motion planning}
The function \texttt{Motion\_Planning} takes among its arguments a
beginning time, an ending time and the number of time intervals.  Its
many outputs are not returned as outputs but stored in global
variables.  The most important is the table \texttt{TTG}.  At each
step time $t_{i}$, the power series expansions $s_{i}$ of the controls
and state variables are stored; \textit{e.g.} for $\alpha$ in
\texttt{TTG$[\alpha,i,0,0]$}. So, they can be used by functions with
names such as \texttt{falpha}, \dots\ that will
compute the value of $\alpha$ at $t_{1}\le t\le t_{i+1}$ using the formula:
$[(t_{i+1}-t)s_{i}(t-t_{i})+(t-t_{i})s_{i+1}(t-t_{i+1})]/(t_{i+1}-t_{i})$
for better precision.

An option calls Maple numerical solver to build numerical integrators
for the full model (stored in \texttt{resudsolve}), using just the
control functions computed with the simplified model, or completing
them with feed-back functions (stored in \texttt{resudsolveB}), that
are described in the next subsection. Then, the function
\texttt{bouclage} that computes the feedback is also called.

An extensive use of the \texttt{subs} Maple function allows to
perform rewriting tasks, replacing in the equations parameters by
their values, as well as already computed state variables.
A basic function \texttt{serpol} (and avatars that apply to both terms
of an equality, list of equalities etc.) computes a power series expansion and
convert it to a polynomial, that is easier to handle for further
computations.

\subsection{Design of the feed-back}\label{subsec:feed-back}

The function \texttt{bouclage} that computes the feedback takes a
single argument that is the fourth linearizing output: $\beta$, $\mu$
or $F$. It return no values, the computed results being stored in
global variables.

To design the feed-back, we consider the linearized system around the
trajectory planned using $\dd x$, $\dd y$, $\dd z$ as flat outputs of
this linear system, completed with $\rd \beta$ or $\rd \mu$, according
to the case, or nothing with the $F$ output. The state functions are
replaced at each step $i$ by its power series expansion at
$t_{i}$. The main idea is to achieve an exponential decrease of
$\updelta x$, $\updelta y$, \dots\ that is the difference between
values $x$, $y$, \dots\ computed by numerical integration using the
full model and the planned values $\tilde x$, $\tilde y$, \dots\ using
the flat parametrization. To be able to correct model errors, we also
need to use the integrals
\begin{equation}
\begin{array}{ll}
I_{1}&=\int_{t_{0}}^{t}\cos(\chi(\tau))\updelta x(\tau)+\sin(\chi(\tau))\updelta y(\tau)\dd\tau;\\
I_{2}&=\int_{t_{0}}^{t}-\sin(\chi(\tau))\updelta
x(\tau)+\cos(\chi(\tau))\updelta y(\tau)\dd \tau;\\
I_{3}&=\int_{t_{0}}^{t}\updelta z(\tau)\dd \tau;\\
I_{4}&=\int_{t_{0}}^{t}\updelta \zeta(\tau)\dd \tau,
\end{array}
\end{equation}
where $t_{0}$ is the initial time of the simulation and $\zeta$ is
$\beta$ or $\mu$ according to our choice of flat outputs.

The algebraic design of the feed-back relies on computations in the
differential module defined by the linearized system at each step time
$t_{i}$, using the analogy between the assumed ``small variations''
$\updelta \xi= \xi-\tilde \xi$ and $\dd \xi$ for any state variable
$\xi$. Each equation $P$ of the system is replaced by its differential
$\sum_{\xi} \partial P/\partial \xi \dd \xi$ and one substitutes to
the $\xi$'s their power series estimation $\tilde \xi$.

Lists of positive real values $\lambda_{i,j}$ having been given, the
feed-back $\updelta
F=c_{1,I_{1}}+\sum_{\xi\in\Xi_{1}\cup\Xi_{2}\cup\Xi_{3}}c_{1,\xi}\updelta_{\xi}$
is set so that $\prod_{k=1}^{3}(\dd/\dd t - \lambda_{1,k})I_{1}$ is
equal to $0$. In the same way, the feed-backs $\updelta
\delta_{\ell}=\sum_{\xi\in\hat\Xi}c_{2,\xi}\updelta_{\xi}$, $\updelta
\delta_{m}=\sum_{\xi\in\hat\Xi}c_{3,\xi}\updelta_{\xi}$ and $\updelta
u_{4}=\sum_{\xi\in\hat\Xi}c_{4,\xi}\updelta_{\xi}$, where
$\hat\Xi=\{I_{1}, \ldots, I_{4}\}\cup\bigcup_{p=1}^{4}\Xi_{p}$, are
computed, so that
$\prod_{k=1}^{5}(\dd/\dd t - \lambda_{2,k})I_{2}$,
$\prod_{k=1}^{5}(\dd/\dd t - \lambda_{3,k})I_{3}$ and
$\prod_{k=1}^{3}(\dd/\dd t - \lambda_{4,k})I_{4}$ are all equal to $0$.

We proceed just as for the motion planning. At each step time $t_{i}$,
an the expressions for $\updelta F$, $\updelta\delta_{\ell}$,
\dots\ are computed and stored in the global array \texttt{TtF}$[i]$,
\texttt{Ttdeltal}$[i]$, \dots\ so that these results can be used by
numerical functions \texttt{ftF}, \texttt{ftdeltal}, \dots\ that achieve fast
numerical computation of the feed-back during the integration.

Under good hypotheses, the $I_{p}$, $1\le p\le 4$ tend to a constant
value, or a slowly varying value, so that their derivatives are $0$,
or small, just as the $\updelta x$, $\updelta y$, $\updelta z$ and
$\updelta \zeta$.  Troubles appear with fast maneuvers and also with
aircrafts like the Twin Otter with generous controls surfaces,
generating greater thrusts. Too big values for the $\lambda_{i,j}$ can
create instabilities, two small values do not manage to keep close to
the planned trajectory. Choices where made with trial and errors, that
sometimes required many interrupted simulations.
\medskip

The choice of $F$ as a flat output just requires minor changes. We
only need to use $I_{1}$, $I_{2}$ and $I_{3}$ and compute the
feed-backs $\updelta \delta_{\ell}$ \dots\ so that
$\prod_{k=1}^{5}(\dd/\dd t - \lambda_{p,k})I_{p}$, for $1\le p\le 3$.

\section{Generalized flatness}\label{sec:generalized}

\subsection{Calibration functions}\label{subsec:calibration}

When the torsion and the curvature of the trajectory are constants,
the values of the controls $F$, $\delta_{l}$, $\delta_{m}$ and
$\delta_{m}$ are constant too. It is then possible to compute them,
just knowing $V$, $\gamma$, $\chi'$ and $\beta$, even for the full
model. They are solutions of a non-linear system, that may be solved
using Newton method. Indeed, looking at the set of
equations~\eqref{eq:shape3}, \eqref{eq:shape4} and the
equations~\eqref{eq::aircraft1:V} and \eqref{eq::aircraft1:gamma}, we
see that for such trajectories, the derivatives in the left members
are equal to $0$. On may add equation \eqref{eq::aircraft1:chi}, for
which the left member $\chi'$ is a constant. We have then $9$
equations between the $13$ unknowns in
$\{V,\gamma,\chi',F\}\cup\Xi_{3}\cup\Xi_{4}\cup\Xi_{5}$. Generically, we
need to fix $4$ values to have local expressions of the $13$
others. We have implemented such functions to compute the angle of
attack $\alpha$, depending of $V$, or to compute stalling speed. They
most of the time only depend of $2$ arguments, instead of $4$, when assuming
$\gamma=\chi'=0$, or just one, when assuming also $\beta=0$.

\subsection{From calibration to time varying controls}

When the control functions are not constant, it remains possible to
evaluate their values with the full system.  The basic idea is to
recompute the trajectory planned with the simplified system, using the
values obtained for $p,q,r,\delta_{\ell},\delta_{m},\delta_{n}$,
instead of $0$. The process can then be iterated, and we can describe
it in the general setting of an almost chained system, such as
\begin{equation}\label{eq:chained-sys}
  \begin{array}{ll}
(Z_{h}',X_{h}') = &G_{h}(Z_{1}, \ldots Z_{h+1}, X_{1}, \ldots,
    X_{h+1})\\
    &+H_{h}(X_{h+2}, \dots, X_{h+\ell_{h}}), 1\le h\le r,
  \end{array}
\end{equation}
with the $\ell_{h}\ge1$, $1\le h\le r$. By convention, $\ell_{h}=1$
means that $H_{h}=0$. The $X_{h}$ form a partition of X, the $Z_{h}$ a
partition of $Z$ and $X\cup Z$ is the set of both state variables and
controls, the distinction being more physical than mathematical.  We
assume that $\sharp X_{h}+\sharp Z_{h}=\sharp X_{h+1}$, $\sharp Z=m$,
the number of controls and $\sharp X_{1}=0$, where $\sharp X_{p}$
denotes the cardinal of $X_{p}$.

If one neglects $H$, or replace in $H$ its arguments by any known
value $\hat X$, the variables in $Z$ are assumed to be flat outputs
for the system. This assumption means that setting
$Z_{h,i}=\zeta_{h,i}(t)$, one can at time $t_{0}$ replace $Z_{h,i}$ in
the equations~\eqref{eq:chained-sys} by a power series development of
$\zeta_{h,i}$ at order $\kappa-h+1$ and compute power series solutions
$\tilde X_{h}$ at order $\kappa-h+1$.  This is assumed to be implemented
in a function \texttt{FlatParametrization}$(t_{0},\kappa,\zeta, \hat
X)$. Using any guessed value $\hat X^{[-1]}$, with $X_{h}^{[-1]}$
known at order $\kappa-h+\ell_{h}$, we can compute an approximation
of the state and control
$$
\hat X^{[0]}:=\hbox{\texttt{FlatParametrization}}(t_{0},\kappa_{0},\zeta,
\hat X^{[-1]}),
$$ where each set $\hat X_{h}$ is computed at order $\kappa_{0}-h+1$.

This may be iterated $J$ times, using $X^{[0]}$, $X^{[1]}$,
\dots\ instead of the guessed value $X^{[-1]}$, as described by the
following process, where the input $v$ denotes the guessed initial
value, $\zeta$ any vector of $m$ functions, $J$ a non-negative integer
and $e$ the wanted order for the output. The order of the output
decreases of $L:=\max_{h=1}^{r}\ell_{h}-1$ at each iteration.
\smallskip

{

\noindent \texttt{GeneralizedFlatParametrization}($v$, $\zeta$,  $J$, $e$)

\noindent $\hat X^{[-1]}:=v\>\hbox{\it (Guessed values)}$;

\noindent $L:=\max_{h=1}^{r}\ell_{h}-1$;

\noindent $\kappa_{0}:=e+r+JL$;

\noindent \textbf{for} j \textbf{from} $0$ \textbf{to} $J$ \textbf{do}

$\hat X^{[j]}:=
\hbox{\texttt{FlatParametrization}}(t_{0},\kappa_{j},\zeta,
\hat X^{[j-1]})$,

$\kappa_{i+1}:=\kappa_{i}-L$;

\noindent \textbf{od};

\noindent \texttt{return} $\hat X^{[J]}$;
\smallskip

}

Returning to the plane model, we have $\sharp X_{1}=0$, $\sharp
X_{2}=\sharp X_{3}=3$ and $\sharp X_{4}=\sharp X_{5}=4$, adding
$F^{(p-3)}$ to $\Xi_{p}$, for $p=4,5$, for consistency with
\eqref{eq:chained-sys}. Furthermore, we have $Z_{1}=\{x,y,z\}$ and
$Z_{3}=\{\xi\}\in\{\alpha,\beta,\mu,F\}$, with
$X_{3}=\{\alpha,\beta,\mu, F\}\setminus\{\xi\}$.

  The only term $H$ is $H_{2}$, that depends of the state variables
  $p$, $q$, $r$ in $X_{4}$ and the controls $\delta_{l}$, $\delta_{m}$
  and $\delta_{n}$ in $X_{5}$. So, $L=2$ in
  our case. This means that with $J$
iterations, we need to start computations with series of order $5+2J$
in oder to get the controls $\delta$ in $X_{5}$ at order $1$.

All the unavoidable accessory tinkerings in the real implementation
would be tedious to detail, but basically, implementing generalized
flat parametrization is an easy task, as we just have to increase the
orders of a known integer and to implement a loop that iterates the
core of the \texttt{Motion\_Planning} function. At iteration $j$, the
series corresponding, \textit{e.g.}, to $\alpha$ is stored in
\texttt{TTG$[\alpha,i,j,0]$}. 

We do not investigate more deeply here the question of the convergence of this
process, beyond the fact that the $H_{h}$ are assumed to be ``small''
and that a limited number of iterations provide good results in the
following examples, all computed with $J=4$.

\section{Examples}\label{sec:examples}

Designing a trajectory that matches actual practice and aircrafts
possibilities by looking at flight instructions books and pilots forums
sure helps. We did not try to use tricks to reduce computation time in
order to get better precision.

\subsection{Single engine}\label{subsec:single_engine}

We model here a Twin Otter that loses an engine, whose power gradually
decreases. We go from equal thrust to total extinction of starboard
engine, setting the value of $\eta=(F_{1}-F_{2})/(F_{1}+F_{2})$, as in
equation~\eqref{eq::TO_single_engine} below. The distance of the
engines to the plane of symmetry of the aircraft has been evaluated to
$9.2$ft.

The rudder must compensate the torque created by a
dissymmetric thrust. With the full model, the rudder also creates a
thrust, that must be compensated by a variation of $\beta$ or $\mu$.
With $\beta=0$ or $\mu=0$, the trajectory planned by the
simplified model is the same. Using here the feed-back for $\beta$,
$\mu$ will change.

{ 
\begin{equation}\label{eq::TO_single_engine}
\begin{array}{ll}
  x&=140{\kts}t;\quad y=0;\quad
  z=0;\quad \mu=0;\\
  \eta&= .5+\frac{\arctan\frac{t-30.}{5.}}{\pi}
\end{array}
\end{equation}
}
The Twin Otter has generous control surfaces, making it highly
manoeuvrable, but meaning a higher contribution of the $\delta_{l}$,
$\delta_{m}$, $\delta_{n}$ to $C_{x}$, $C_{y}$ and $C_{z}$. We borrow
with some adaptations the values of the $\lambda_{i,j}$ suggested by
Martin~\cite{MartinPHD}:
$\lambda_{1,1}=1.$, $\lambda_{1,2}=2.$, $\lambda_{1,3}=
3.$, $\lambda_{2,1}= 1.$, $\lambda_{2,2}=  1.$, $\lambda_{2,3}=
1.$, $\lambda_{2,4}=  2.$, $\lambda_{2,5}=  3.$,
$\lambda_{3,1}=  1.5$, $\lambda_{3,2}=  1.5$, $\lambda_{3,3}=  1.5$,
$\lambda_{3,4}=  3.$, $\lambda_{4,5}=  4.$, $\lambda_{4,1}=  1.$,
$\lambda_{4,2}=  2.$, $\lambda_{4,3}=  3$.

The variations of $\mu$ remains little,
in accordance with the reported ability of the T-O to fly with a
single engine (Lecarme~\cite{Lecarme1966}).

\begin{figure}[ht]\label{fig::GTM_se}
  \caption{Twin Otter loosing one engine, with $\beta=0$.}
  \hbox to \hsize{\includegraphics[width=6cm]{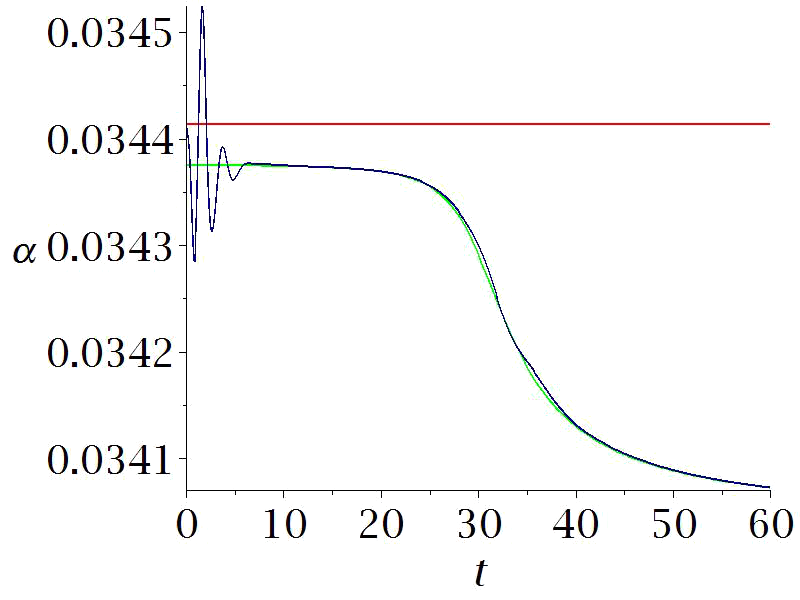}\hfill
    \includegraphics[width=6cm]{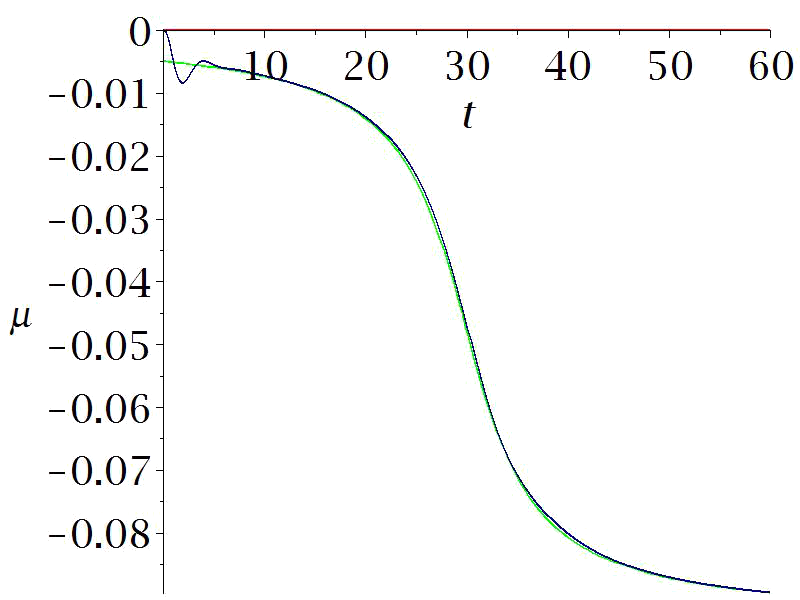}}
  
\vbox{\small The flatness planned curve is in red, the integration
  with feed-back in darkblue and the generalized flatness curve in
  green.}
\Description{The planned curve is in red, the integration
  with feed-back in darkblue and the generalized flatness curve in
  green.}
\end{figure}
We see that the integrated curves converge to the curves planned by
generalized flatness, after initial oscillations, which already shows
that this prediction is meaningful. The total computation time for the
flat and generalized parametrization is $1279$sec. The numerical simulation
takes $76$sec.

\subsection{Forward slip}\label{subsec:forward_slip}

This maneuver may be used for emergency landing, when an aircraft that
has lost all engines comes near the landing strip too high or too
fast. A way to decrease speed and altitude is to increase $\beta$ and
$\mu$ in opposite ways, creating deceleration when aerobrakes are
unusable. It is in general used for small aircrafts, but there is a
successful example of an emergency landing with an airliner, at the
former air force basis of Gimli, Manitoba, in 1983~\cite{Gimli}. Here
we used a calibration function to guess initial values and non zero
values for the controls, close to the mean speed and flight path angle
of our trajectory.

The following table shows constant values for straight line
trajectories, depending on $\alpha$ and $\beta$, for both the real and
the simplified models with
$(p,q,r,\delta_{l},\delta_{m},\delta_{n})=(0,0,0,0,0,0)$.
{\footnotesize
$$
\begin{array}{|l|l|l|l|l|l|l|l|l|}
\hline
\hbox{Model}&\alpha&\beta&\gamma&\mu&V&\delta_{l}&\delta_{m}&\delta_{n}\\
\hline
\hbox{Simple}&0.15&0.&-0.1187&0.&29.8996&0.&0.&0.\\
\hline
\hbox{Real}&0.15&0.&-0.1190&0.&30.3053&0.&-0.0490&0.\\
\hline
\hbox{Simple}&0.15&0.2&-0.1650&0.2409&29.3672&0.&0.&0.\\
\hline
\hbox{Real}&0.15&0.2&-0.1470&0.1345&30.1114&-0.1880&-0.0490&0.3305\\
\hline
\hbox{Simple}&0.15&0.35&-0.2508&0.3899&28.4019&0.&0.&0.\\
\hline
\hbox{Real}&0.15&0.35&-0.2027&.2250&29.7171&-0.3316&-0.0490&0.5690\\
\hline
\end{array}
$$ } For our simulation, we have chosen $\alpha=0.15$ and $\beta=0.35$
as reference values to set the controls. To fix ideas, the speed
values for such a $0.055$ scale model must be divided by $0.055^{0.5}$
to get full scale values, which means $456.1709$km/h for the total
speed.  Here are the flat output trajectories and feed-back
parameters.

\begin{equation}\label{eq::GTM_FS}
\begin{array}{ll}
  x&=29.10852587t+50\sin(t/60.);\\
  y&=60\cos(t/100.+2.);\\
  z&=-1000+5.983293200t+70\sin(t/70.));\\
\lambda_{i,j}&=0.5
\end{array}
\end{equation}

\begin{figure}[ht]\label{fig::GTM_FS}\caption{Forward slip with the GTM}
  \hbox to \hsize{
  \includegraphics[width=6cm]{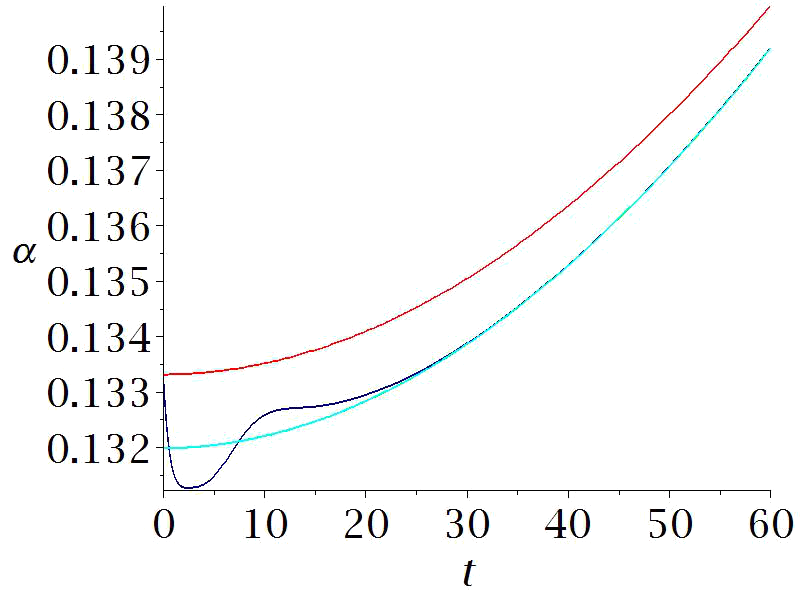}\hfill
  \includegraphics[width=6cm]{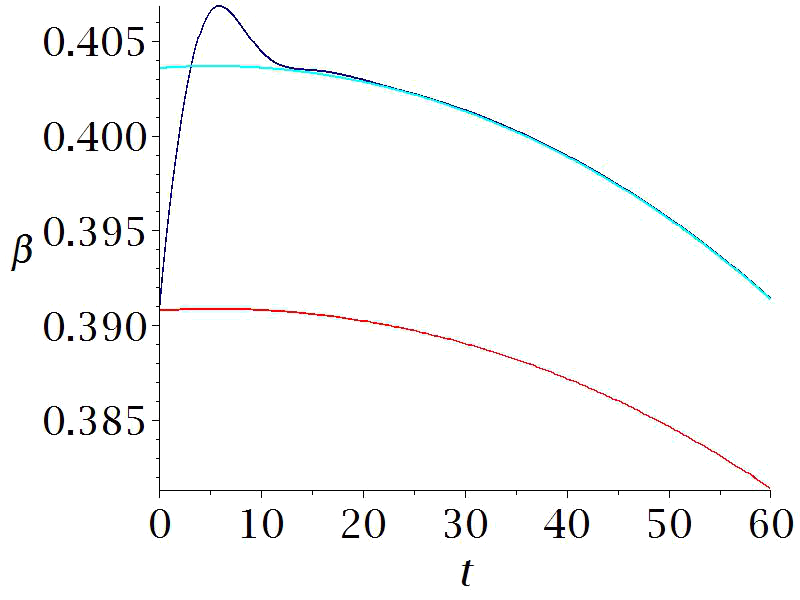}}
  \hbox to \hsize{
  \includegraphics[width=6cm]{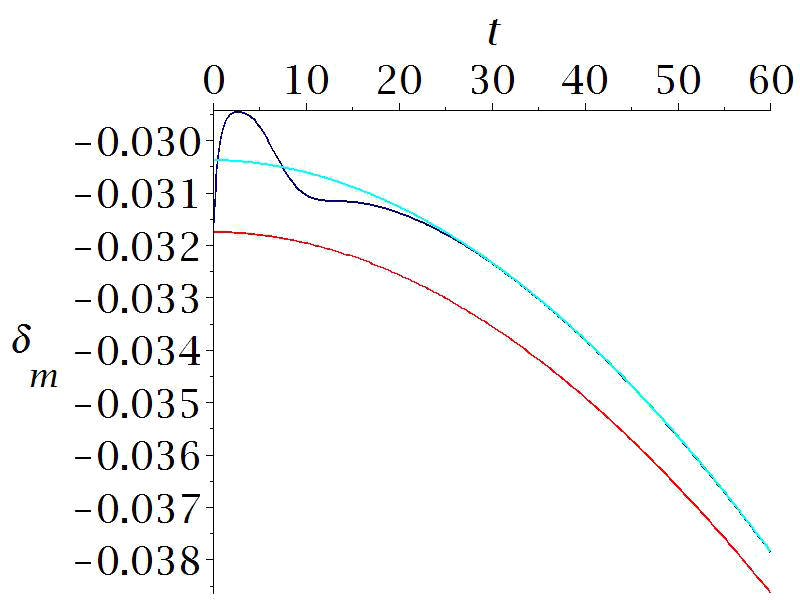}\hfill
  \includegraphics[width=6cm]{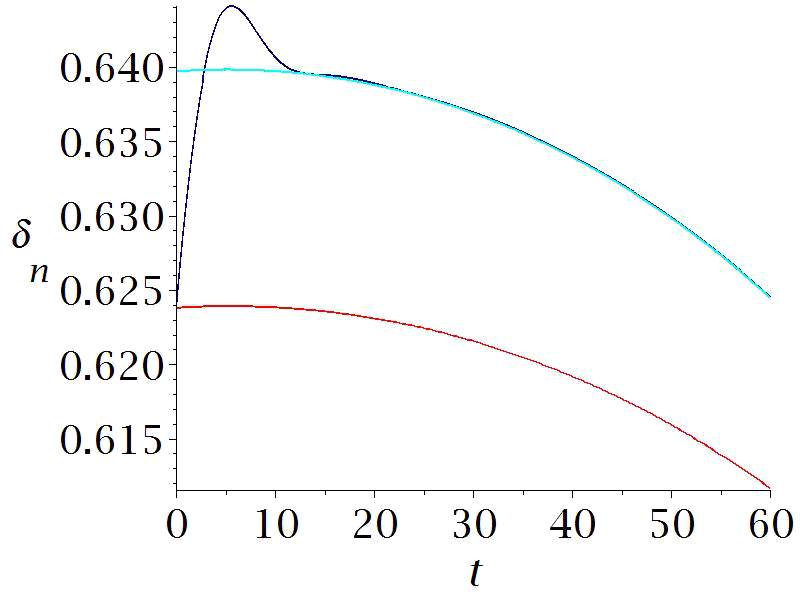}}
\vbox{\small The flatness planned curve is in red, the integration with
  feed-back in dark blue and the generalized flatness curves in green.
 The curve in cyan is the integration with the generalized
    flatness planned controls and without feed-back.}
\Description{The flatness planned curve is in red, the integration
  with feed-back in dark blue and the generalized flatness curves in
  green. The curve in cyan is the integration with the generalized
    flatness planned controls and without feed-back.}
\end{figure}

Again, the feed-back allows the integrated value to converge to the
curve planned by generalized flatness with good precision, after
initial oscillations. The curves $\delta_{l}$ and $\delta_{m}$
actually show $\delta_{m}+\updelta\delta_{m}$ and
$\delta_{n}+\updelta\delta_{n}$, including feed-back. We have
included here the integration of the general system, with we initial
values and control coming from generalized flatness. The coincidence is so
good that the generalized flatness planned curves in green are
covered by the curve in cyan provided by the integration.

\subsection{Aileron roll and parabolic
  flight}\label{subsec:parabolic_roll}

Here, we investigate a limit case with rapid changes. The trajectory
is parabolic with acceleration $g$, so the flat outputs with $\beta$
is unusable. We use $\mu$, setting $\mu=\pi/2t$. A
fighter would have been more credible, but we could only make the
feed-back work with the GTM. The horizontal speed is
$100\hbox{km}/\hbox{h}$. 

\begin{figure}[ht]\label{fig::GTM_para_roll}
  \caption{Aileron roll and parabolic flight with the GTM.}
  \hbox to \hsize{
    \includegraphics[width=6cm]{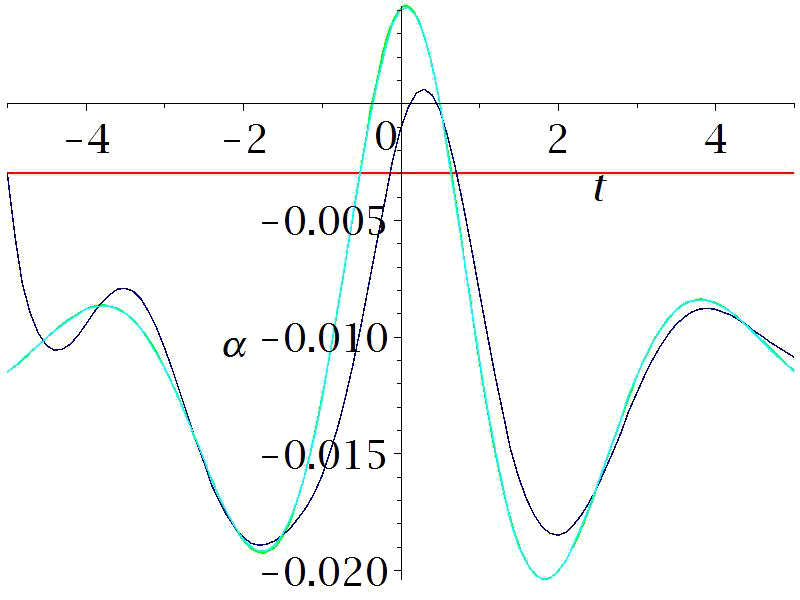}\hfill
    \includegraphics[width=6cm]{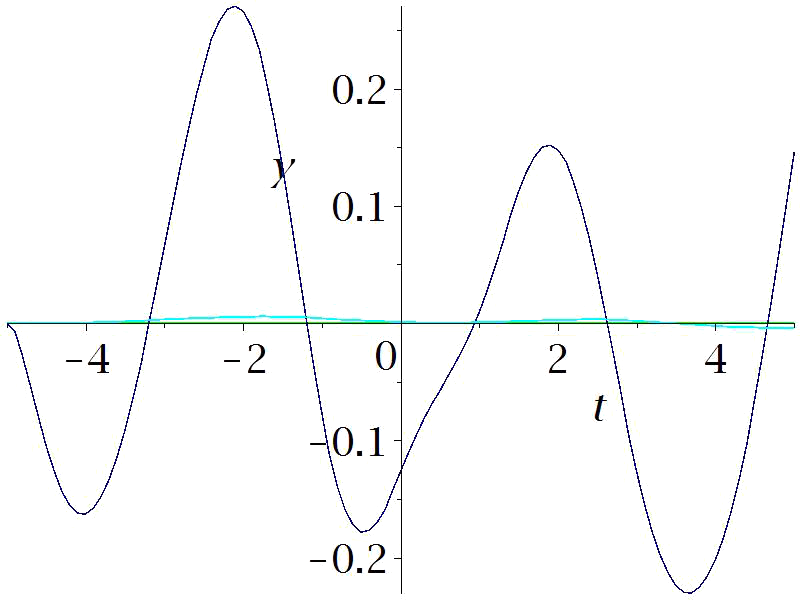}}
  \hbox to \hsize{
    \includegraphics[width=6cm]{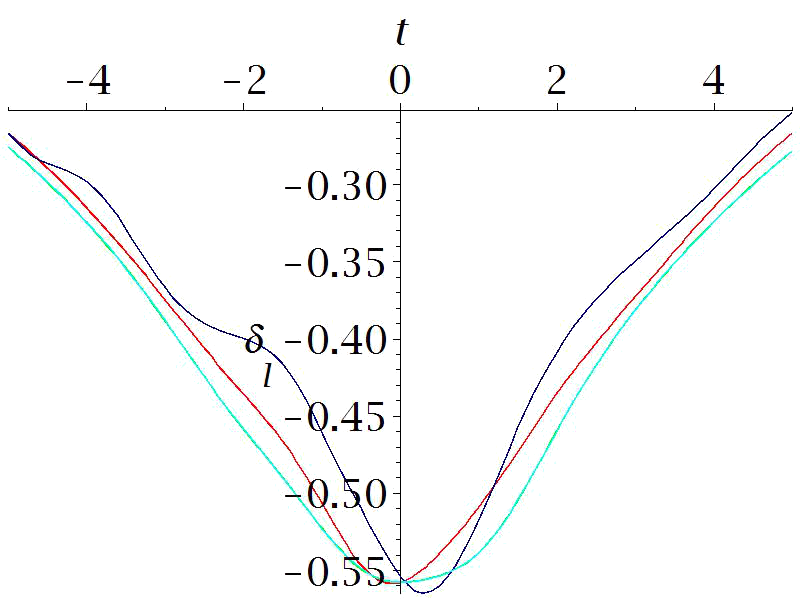}\hfill
    \includegraphics[width=6cm]{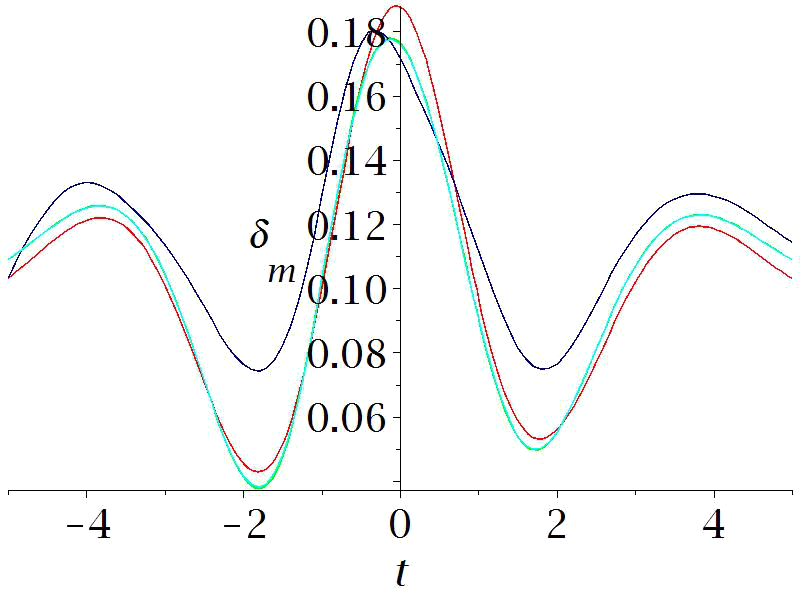}}
    
  \vbox{\small The flatness planned curve is in red, the integration with
    feed-back in dark blue and the generalized flatness curve in
    green. The curve in cyan is the integration with the generalized
    flatness planned controls and without feed-back.}
  \Description{The
    flatness planned curve is in red, the integration with feed-back in
    dark blue and the generalized flatness curve in green. The curve
    in cyan is the integration with the generalized flatness planned
    controls without feed-back, that match very well the generalized
    flatness prediction.}
\end{figure}
\vfill\eject

We see that the feed-back permits to follow the generalized flatness
planned curve, but things are moving too fast to keep always the two
curves close. The integration in cyan with the generalized
flatness planned control, without feed-back, remains very close to the
prediction, which confirms that the generalized flatness
parametrization is a good approximation of a solution of the real
system. \textit{E.g.}, a small discrepancy of about $0.5$cm, is observed
for $y$ at $t=5.$, one of the only state function for which the curve in green
appears bellow the cyan one. The computation time is $647$sec for the
motion planning and $402$sec for the simulation.

To better appreaciate the convergence of the generalized flatness
loop, we have computed the values for the controls $F$, $\delta_{l}$,
$\delta_{m}$ and $\delta_{n}$ at $t=-1.9$ a time  for which
the differences with the plain flatness values are much
appreciable. They
are given in the table bellow.

{\footnotesize

$$
\begin{array}{|l|l|l|l|l|l|l|l|l|}
\hline
 &J=0&J=1&J=2&J=3&J=4&J=5&J=6&J=7\\
\hline
F        &-2.36&8.40 &8.56&8.610&8.624&8.628&8.6304&8.6309\\
\hline
\delta_{l}&-0.44&-0.45&-0.462&-0.4642&-0.4647&-0.4648&-0.46493&-0.464918\\
\hline
\delta_{m}&0.04&0.04&0.039&0.0389&0.0387&0.03872&0.038730&0.038731\\ 
\hline
\delta_{n}&0.05&0.07&0.085&0.0871&0.087&0.08800&0.087997&0.0880978\\
\hline
\end{array}
$$

}

The theoretical study of convergence is of course of a great interest, but it is
known that such a property is not mandatory for
applications. \textit{E.g.}, some divergent series, using smallest term
trunctation, can provide accurate and fast
computations. See~\cite{Ramis1993}. 

\section{Generalized flatness from the theoretical
  standpoint}\label{sec:theory} 

The flat
parametrization only involves a finite number of derivatives, which is
the basis of all known necessary conditions of flatness
(see~\cite{Sluis1993,Rouchon1994,Ollivier1998}). We have seen that our
motion planning is a limit that potentially involves an infinite number
of derivatives, as the evaluation for the controls $\delta$ at step
$j+1$ depends on the second derivative of their evaluation at step
$j$.  This gives some credibility to a folkloric conjecture, claiming
that \textit{all controllable systems are flat if functions of an
  infinite number of derivatives are allowed}. We propose some
elements of interpretation in the linear case.

We may indeed consider the simple system
$x'=y+\epsilon y'$. When $\epsilon$ is $0$, $x$ is a flat output. For
$\epsilon>0$, we may choose
$\zeta_{\epsilon}:=x-\epsilon y$. However, we can keep $x$ as a \emph{generalized flat output}. Indeed,
one may write
$y=\sum_{i\in\N}(-1)^{i}\epsilon^{i}(\dd/\dd t)^{i}x$. This series
will converge if $x$ is analytic with a convergence radius greater
that $1/\epsilon$. Moreover, if there exists a linear operator $L$ in
$\R[\dd/\dd t]$ such that $Lx=0$ and $1+\epsilon \dd/\dd t$, as well as
$\dd/\dd t$, are not a
factors of $L$, then there exists $M$ and $N$ such that
$ML+N(1+\epsilon \dd/\dd t)=1$, so that $y=Nx'$. Taking for $L$ the
sequence $(\dd/\dd t)^{i}$, the sum that gives the value of $y$
becomes trivially finite. This situation is close to our considerations about
calibration in subsec.~\ref{subsec:calibration}.  But this can work also
with any operator $\prod_{i=1}^{k}(\dd/\dd t - \lambda_{i})^{i}$, such as
those that we met for designing feed-backs in subsec.~\ref{subsec:feed-back}. 
\vfill\eject

\section*{Conclusion}

We have seen how computer algebra may help to investigate
the validity of some simplifications required to reduce to a flat
model. Although we could rely on very classical algorithmic tools, some
investment have been required to work out for our experiments an
implementation with acceptable computation times. One also need a joint
use of symbolic and numeric computations.

A slight modification of the code used with the simplified flat model
have made possible the direct computation of an accurate motion
planning for the original non flat system, an observation that cannot
be a mere artefact and so requires a theoretical explanation.

One cannot predict if this notion of generalized flatness will have
actual applications. The theoretical difficulties are also unkown, but
the unanswered problems related to flatness show that limited
theoretical knowledge is not an obstacle to applicability, as long as
computations are fast and results reliable. The complexity of the
model used here could justify some optimism for computational success
with much simpler examples, such as the car with two deported
trailers, known not to be flat~\cite{Rouchon1993}.

Those investigations include an algorithmic aspect. \textit{E.g.}, one
may ask whether is it possible to compute the generalized
parametrization in a faster way, using some kind of Newton method,
which could also help to investigate the convergence of the
process. So, even if the applicability should be limited,
computational issues may remain of some interest.  
\bigskip

\noindent\textbf{Thanks} To Yirmeyahu J. Kaminski, Jean Lévine and
anonymous referees for their patience, rereading and suggestions.

\bibliographystyle{amsplain}

\bibliography{Ollivier-ISSAC.bib}

\providecommand{\bysame}{\leavevmode\hbox to3em{\hrulefill}\thinspace}
\providecommand{\MR}{\relax\ifhmode\unskip\space\fi MR }
\providecommand{\MRhref}[2]{%
  \href{http://www.ams.org/mathscinet-getitem?mr=#1}{#2}
}
\providecommand{\href}[2]{#2}
\begin{thebibliography}{10}

\bibitem{aecf-2017-livre}
Alin Bostan, Frédéric Chyzak, Marc Giusti, Romain Lebreton, Grégoire Lecerf,
  Bruno Salvy, and Éric Schost, \emph{Algorithmes efficaces en calcul formel},
  Frédéric Chyzak (auto-édit.), Palaiseau, September 2017 (french), 686
  pages. Printed by CreateSpace. Also available in electronic version.

\bibitem{FLMR95}
M.~Fliess, J.~L\'{e}vine, Ph. Martin, and P.~Rouchon, \emph{Flatness and defect
  of non-linear systems: introduction theory and examples}, Int. Journal of
  Control \textbf{61} (1995), no.~6, 1327--1361.

\bibitem{FLMR99}
\bysame, \emph{A {L}ie-{B}\"{a}cklund approach to equivalence and flatness of
  nonlinear systems}, IEEE Trans. Automatic Control \textbf{44} (1999), no.~5,
  922--937.

\bibitem{Grauer-Morelli2014}
Jared~A. Grauer and Eugene~A. Morelli, \emph{A generic nonlinear aerodynamic
  model for aircraft}, AIAA Atmospheric Flight Mechanics Conference, AIAA,
  2014.

\bibitem{Hueschen}
Richard~M. Hueschen, \emph{Development of the transport class model (tcm)
  aircraft simulation from a sub-scale generic transport model (gtm)
  simulation}, Tech. Report NASA/TM–2011-217169, NASA, 2011.

\bibitem{KLO18}
Y.~Kaminski, J.~L\'evine, and F.~Ollivier, \emph{Intrinsic and apparent
  singularities in differentially flat systems, and application to global
  motion planning}, Systems \& Control Letters \textbf{113} (2018), 117--124.

\bibitem{KLO20}
\bysame, \emph{On singularities of flat affine systems with n states and $n -
  1$ controls}, International Journal of Robust and Nonlinear Control
  \textbf{30} (2020), no.~9, 3547--3565.

\bibitem{Vinogradov1986}
V.V. Krasil'shchik, V.V. Lychagin, and A.M. Vinogradov, \emph{Geometry of jet
  spaces and nonlinear partial differential equations}, Gordon and Breach, New
  York, 1986.

\bibitem{Lecarme1966}
J.~Lecarme, \emph{Lignes de vol, le de havilland dhc-6 twin otter}, Aviation
  Magazine (1966), no.~449.

\bibitem{Levine09}
J.~L\'{e}vine, \emph{Analysis and control of nonlinear systems: A
  flatness-based approach}, Mathematical Engineering, Springer, Dordrecht,
  Heidelberg, London, New-York, 2009.

\bibitem{Levine11}
\bysame, \emph{On necessary and sufficient conditions for differential
  flatness}, Applicable Algebra in Engineering, Communication and Computing
  \textbf{22} (2011), no.~1, 47--90.

\bibitem{Gimli}
George~H. Lockwood, \emph{Final report of the board of inquiry into air canada
  boeing 767 c-gaun accident --- gimli, manitoba, july 23, 1983}, Tech. report,
  Minister of Supply and Services Canada, 1985.

\bibitem{Lu-Turgolu-2018}
Long~K. Lu and Kamran Turkoglu, \emph{Adaptive differential thrust methodology
  for lateral/directional stability of an aircraft with a completely damaged
  vertical stabilizer}, International Journal of Aerospace Engineering
  \textbf{218} (2018).

\bibitem{MartinPHD}
P.~Martin, \emph{Contribution \`{a} l'\'{e}tude des syst\`{e}mes
  diff\'{e}rentiellement plats}, Ph.D. thesis, Ecole Nationale Sup\'{e}rieure
  des Mines de Paris, Paris, France, 1992.

\bibitem{Martin1996}
Philippe Martin, \emph{Aircraft control using flatness}, CESA'96 - Symposium on
  Control, Optimization and Supervision (Lille, France), IMACS/IEEE-SMC
  Multiconference, 1996, pp.~194--1999.

\bibitem{McLean}
Donald McLean, \emph{Automated flight control systems}, Prentice Hall, New
  York, 1990.

\bibitem{Ollivier1998}
François Ollivier, \emph{Une réponse négative au problème de lüroth
  différentiel en dimension 2}, C. R. Acad. Sci. Paris \textbf{327} (1998),
  no.~10, 881--886.

\bibitem{Ramis1993}
J.P. Ramis, \emph{Séries divergentes et théories asymptotiques}, Société
  Mathématique de France, Marseille, 1993.

\bibitem{Ritt50}
J.F. Ritt, \emph{Differential algebra}, American Mathematical Society,
  Providence, Rhodes Island, 1950.

\bibitem{Rouchon1993}
P.~Rouchon, M.~Fliess, J.~Levine, and P.~Martin, \emph{Flatness, motion
  planning and trailer systems}, Proceedings of 32nd IEEE Conference on
  Decision and Control, IEEE, 1993, pp.~2700--2705 vol.3.

\bibitem{Rouchon1994}
Pierre Rouchon, \emph{Necessary condition and genericity of dynamic feedback
  linearization}, Journal of Mathematical Systems Estimation and Control
  \textbf{4} (1994), no.~2, 1--14.

\bibitem{Sluis1993}
Willem~M. Sluis, \emph{A necessary condition for dynamic feedback
  linearization}, Systems \&\ Control Letters \textbf{21} (1993), 277--283.

\bibitem{Zharinov1996}
Victor~V. Zharinov, \emph{Geometrical aspects of partial differential
  equations}, Series on Soviet and East European Mathematics, World Scientific,
  Singapore, 1992.

\end{thebibliography}

\end{document}